\documentclass[a4paper,12pt]{article}
\usepackage[utf8]{inputenc}
\usepackage{amsmath, amssymb, amsthm, amscd}
\usepackage{mathrsfs,mathtools}
\usepackage[text={5.8in,8.8in},centering]{geometry}
\usepackage{color}

\theoremstyle{plain}
\newtheorem{theorem}{Theorem}
\newtheorem{proposition}{Proposition}
\newtheorem{lemma}{Lemma}

\theoremstyle{definition}

\theoremstyle{remark}
\newtheorem{remark}{Remark}[section]

\newcommand{\Res}{\mathrm{Res}}
\def\bbbone{{\mathchoice {\rm 1\mskip-4mu l} {\rm 1\mskip-4mu l}
{\rm 1\mskip-4.5mu l} {\rm 1\mskip-5mu l}}}
\def\one{\bbbone}
\renewcommand{\i}{\mathrm{i}}
\newcommand{\e}{\mathrm{e}}

\newcommand{\pder}{\partial}

\renewcommand{\Re}{\mathrm{Re}\,}

\newcommand{\Tr}{\mathrm{Tr}\,}

\newcommand{\bbN}{\mathbb{N}}

\newcommand{\bbC}{\mathbb{C}}
\newcommand{\bbR}{\mathbb{R}}

\newcommand{\Op}{\mathrm{Op}}
\begin{document}

\title{On the Weyl symbol of the resolvent
  \\ of the harmonic oscillator}
\author{
	Jan Derezi\'{n}ski,\footnote{The financial support of the National Science
		Center, Poland, under the grant UMO-2014/15/B/ST1/00126, is gratefully
		acknowledged.}
	\hskip 3ex
	Maciej Karczmarczyk\footnotemark[\value{footnote}]
	\\
	Department of Mathematical Methods in Physics, Faculty of Physics\\
	University of Warsaw,  Pasteura 5, 02-093, Warszawa, Poland\\
	email: jan.derezinski@fuw.edu.pl\\
	email: maciej.karczmarczyk@fuw.edu.pl}
\maketitle
\begin{abstract}We compute the Weyl symbol of the resolvent of the harmonic oscillator and study its properties.\end{abstract}


\section{Introduction}

Throughout our paper, by  the {\em (quantum) harmonic oscillator}, we mean the self-adjoint operator on $L^2(\bbR^d)$ defined as
\begin{equation}
  H:=-\Delta+x^2,\end{equation}
where $\Delta$ is the Laplacian and $x^2=\sum_{i=1}^dx_i^2$.  The spectrum of $H$ is  \[\{d,d+2,d+4,\dots\}.\]
  The central object of our paper is the resolvent of $H$, that is,
the operator  $(H-z)^{-1}$ defined
  for $z$ outside of the   spectrum of $H$.

  Another central concept of our paper is the {\em Weyl symbol} of an operator. Weyl symbol is a very natural  parametrization of operators on $L^2(\bbR^d)$, extensively used in PDE's. It is plays also an important role in foundations of quantum mechanics, where it is usually called the {\em Wigner function}, and can be used to express the semiclassical limit.  For an introduction to the Weyl quantization we refer the reader to Chap. 5 of \cite{BS}, Chap. XVIII of \cite{horm3}, \cite{mar}, \cite{zw} or 
  Sec. 8.3 of \cite{jd-cg}.

It is a well-known fact that the Weyl symbol of a function of $H$
is an another function (or distribution) of $x^2+p^2$.
This fact was e.g. shown and discussed in
\cite{jd-weyl}, as well as in the recent paper \cite{toft}.
Therefore, the Weyl symbol of  $(H-z)^{-1}$
can be written as $ F_{d,z}(x^2+p^2)$ for some $F_{d,z}$.

Our paper
is devoted to a  study of the properties of  $F_{d,z}$.
In particular, we give a few explicit expressions for $F_{d,z}$---we present an integral formula, a power series expansion and an expression in terms of confluent-type functions. We also provide some estimates on the  derivatives
 of  $F_{d,z}$. Some of our formulas simplify in the case $z=0$, that is, for  the inverse of the harmonic oscillator. In particular, the Weyl symbol of the inverse can be expressed in terms of Bessel-type functions.

 We find it interesting and potentially useful that the Weyl symbol of the
 resolvent of the harmonic oscillator has an explicit description. With our formulas we are able to study its properties, deriving in particular  rather precise bounds on its derivatives.

 In the literature we have not seen a study of the Weyl symbol of the resolvent of the harmonic oscillator except for a recent paper \cite{toft},
  devoted to the inverse of the harmonic oscillator $H^{-1}$.
 \cite{toft} contains a formula for the Weyl symbol of $H^{-1}$ in terms of a power series. It also proves that its derivatives satisfy some estimates. The authors of \cite{toft} call them
\emph{ Gelfand-Shilov bounds}.

 The results of our paper are stronger than those of \cite{toft}.
First, we consider the more general case of the resolvent $(H-z)^{-1}$, whereas \cite{toft} is restricted to $z=0$. Second, our explicit representation in terms of Bessel-type function and in terms of an integral representation is absent in \cite{toft}. Third, our bounds on the derivatives easily imply  those proven in \cite{toft}.

An interesting discussion of Weyl quantization of spherically symmetric symbols is contained in a recent paper of Unterberger \cite{unterb2}. That paper contains in particular a formula for the symbol of the $n$th spectral projection of the 
harmonic oscillator. We give an alternative derivation of Unterberger's formula, using our results about the resolvent as the starting point.

 \section{Weyl quantization}

 Let us recall the definition of the Weyl quantization, following
 e.g.  \cite{jd-cg}, \cite{zw} or \cite{unterb2}. If $a$ is a  distribution in ${\mathcal S}'(\bbR^d\oplus\bbR^d)$, then its {\em Weyl quantization}
                is defined to be the operator $\Op(a)$
                from ${\mathcal S}(\bbR^d)$
                to ${\mathcal S}'(\bbR^d)$ such that
\begin{equation}
\Op(a)\Phi(x) = (2\pi)^{-d}\int a\Big(\frac{x+y}{2},\,p\Big)\e^{\i p(x-y)}\Phi(y)\,\mathrm{d}p\mathrm{d}y.
\end{equation}
The distribution $a$ is then called the {\em symbol} of the operator $\Op(a)$.

Let us remark that in the literature one can find various classes of symbols. 
${\mathcal S}'(\bbR^d\oplus\bbR^d)$, that we use in our definition, is broad enough for our purposes.  As noted in \cite{toft}, one can define the Weyl quantization on
more general class of symbols: e.g. on the dual of the so-called Gelfand-Shilov space.


The usual product of operators corresponds on the level of symbols to the so-called star product $*$ (sometimes called the {\em Moyal star}).
That means, if
\begin{equation}
(a*b)(x,\,p) :=  \e^{\frac{\i}{2}(\pder_{x_1}\pder_{p_2}-\pder_{p_1}\pder_{x_2})}a(x_1,\, p_1)b(x_2,\, p_2) \Big|_{\begin{subarray}{l}x:=x_1=x_2 \\ p:=p_1=p_2    \end{subarray}},
\label{moyal}\end{equation}
then  $\Op(a)\Op(b) = \Op(a*b)$.

\section{The symbol of the resolvent}
For $z\in\bbC$ outside of the spectrum of $H$,
 $a_{d,z}$ will denote the symbol of the harmonic oscillator, that is,
\begin{equation}
  \Op(a_{d,z})=(H-z)^{-1}.\end{equation}
As discussed in the introduction, we can then define $F_{d,z}$ by
\begin{equation}F_{d,z}(x^2+p^2)=a_{d,z}(x,p).\end{equation}

	Some properties of $F_{d,z}$ can be derived in an easy way with use of the integral representation given by
	\begin{theorem}\label{integral-rep-theorem}
For $\Re(z)<d$, the following formula holds:
		\begin{equation}\label{integral-rep-of-res}
	F_{d,z}(\rho) =  \int\limits_0^1(1-s)^{\frac{d-z}{2}-1}(1+s)^{\frac{d+z}{2}-1}\e^{-s\rho}\,\mathrm{d} s.
		\end{equation}
	\end{theorem}
	\begin{proof}
It is well known that
\begin{equation}\label{unterberger-eq}
		\e^{-tH} = \Op\Big((\cosh t)^{-d}\e^{-\mathrm{tgh}\, t \,(x^2+p^2)}\Big).
		\end{equation}
	(see e.g. \cite{unterb1,jd-cg}).
Hence,
		\begin{align}
		  (H-z)^{-1}&
 = \int\limits_{0}^{\infty}\e^{-tH+tz}\,\mathrm{d} t\\
&
                  =\Op\Big(\int\limits_{0}^{\infty}(\cosh t)^{-d}\e^{-\mathrm{tgh}\, t \,(x^2+p^2)}\e^{tz}\,\mathrm{d} t\Big)\\
                 &= \Op\Big(\int\limits_{0}^{1}(1-s^2)^{\frac{d}{2}-1}\e^{-s(x^2+p^2)}\e^{z\,\mathrm{artgh}\, s}\,\mathrm{d} s\Big),
		\end{align}
		where at the end we made the  substitution $\mathrm{tgh}\, t = s$.
Using $\mathrm{artgh}\,s = \frac{1}{2}\ln\frac{1+s}{1-s}$ and the linearity of the quantization, we obtain the final result. 		
	\end{proof}
	
        \begin{theorem}
          The function $F_{d,z}$ is entire analytic. It can be written as
        \begin{equation}
          F_{d,z}(\rho) = \sum\limits_{k=0}^\infty c_{k} \frac{\rho^k}{k!},\end{equation}
where            \begin{equation}
	      c_k = \frac{\Gamma(k+1)\Gamma(\frac{d-z}{2})}{\Gamma(k+1+\frac{d-z}{2})}\,_2F_1\big(1-\frac{d+z}{2},\,k+1;\,k+1+\frac{d-z}{2};\,-1\big),	\end{equation}
            where $\,_2F_1(a,\,b;\,c;\,z) = \sum\limits_{k=0}^{\infty}\frac{(a)_k(b)_k}{(c)_k}\frac{z^k}{k!}$ is the  hypergeometric function.
	\end{theorem}

          \begin{proof}
            Entire analyticity of $F_{d,z}$ is obvious from
            (\ref{integral-rep-of-res}).
            
            Of course, $c_{k} = \pder_\rho^k F_{d,z}(0).$
            Differentiating the integral representation \eqref{integral-rep-of-res} and putting $\rho=0$ results in 
	\begin{equation}\label{ck-integral-rep}
		c_k = \int\limits_0^1(1-s)^{\frac{d-z}{2}-1}(1+s)^{\frac{d+z}{2}-1}s^k\,\mathrm{d}s.
	\end{equation}        
        Then we apply Euler's formula, 
        see e.g. \cite{der-hyper}, Subsection 3.7.6,
	\begin{align}
		&\frac{\Gamma(b)\Gamma(c-b)}{\Gamma(c)}\,_2F_1(a,\,b;\,c;\,z)\\ = &\int\limits_0^1 s^{b-1}(1-s)^{c-b-1}(1-zs)^{-a}\,\mathrm{d}s,\mbox{ for }\Re(c)>\Re(b)>0.
	\end{align}
	 \end{proof}

          \begin{theorem} For $\rho\to\infty$, we have
            \begin{equation}
              F_{d,z}(\rho)=\frac{1}{\rho}+O\Big(\frac{1}{\rho^2}\Big).
            \end{equation}
            More generally, there exist $d_1:=1,d_2,d_3,\dots$ such that for any $n$
            \begin{equation}
              F_{d,z}(\rho)=\sum_{j=1}^{n-1}\frac{d_j}{\rho^j}+O\Big(\frac{1}{\rho^n}\Big).
            \end{equation}
          \end{theorem}

          \proof
          For $0<a<1$, we write
          \begin{align}
                     \label{toft1}            F_{d,z}(\rho)=&\frac{(-1)^n}{\rho^n}\int_0^a
    (1-s)^{\frac{d-z}{2}-1}(1+s)^{\frac{d+z}{2}-1}\partial_s^n\e^{-s\rho}\,\mathrm{d} s\\
                      &+\int_a^1
 (1-s)^{\frac{d-z}{2}-1}(1+s)^{\frac{d+z}{2}-1}\e^{-s\rho}\,\mathrm{d} s.\label{toft2}
          \end{align}
          (\ref{toft2}) is $O(\e^{-a\rho})$. We integrate (\ref{toft1}) $n$ times by parts. The boundary terms at $s=0$ have the form $\frac{d_j}{\rho^j}$ and  the boundary terms at $s=a$ are $O(\frac{\e^{-a\rho}}{\rho^j})$,
          $j=1,\dots,n$. The remaining integral is $O(\frac{1}{\rho^n})$. \qed

          \section{Confluent-type functions
            and  the harmonic oscillator}\label{sec-differential-res}

The harmonic oscillator is closely related to the confluent equation. We devote this section to this relationship.

First, recall that the \emph{confluent differential operator} is defined as 
	\begin{equation}
		x\pder_x^2 + (c-x)\pder_x - a,\label{conf}
	\end{equation}
	for $a,c\in\bbC$. Among functions annihilated by (\ref{conf}),
        two are distinguished:
	\begin{itemize}
		\item The \emph{confluent}, ${}_1F_1$ or \emph{Kummer's} \emph{function}:
			\begin{equation}
				M(a,c;x):={}_1F_1(a;c;x) = \sum\limits_{k=0}^{\infty}\frac{(a)_k}{(c)_k}\frac{x^k}{k!}.
			\end{equation}
			It is the only solution behaving as $1$ in the vicinity of $0$.
		\item \emph{Tricomi's} \emph{function}:
			\begin{equation}
				U(a,c;x) := x^{-a}{}_2F_0(a;a-c+1;-; -x^{-1}).
			\end{equation}
	Tricomi's function is the only solution having the asymptotic behaviour
	$				U(a,c;z) \sim  z^{-a} $
			at infinity.
	\end{itemize}
        If $a$ is a nonpositive integer, then both Kummer's and Tricomi's functions are proportional to Laguerre polynomials:
        \begin{align}
          L_n^\alpha(x)&=\frac{(\alpha+1)_n}{n!}M(-n;\alpha+1;x)\label{laguerre1}
          \\
          &=\frac{(-1)^n}{n!}U(-n;\alpha+1;x).\label{laguerre2}
            \end{align}

        We will need Green's function of the confluent operator, that is, the integral kernel of its inverse $R(a;b;x,y)$. It should satisfy
        \begin{align}
          \big(x\pder_x^2 + (c-x)\pder_x - a\big)R(a;b;x,y)&=\delta(x-y),\\
                              R(a;b;x,y)\big(y\pder_y^2 + (c-y)\pder_y - a\big)&=\delta(x-y).
	                      \end{align}
                                	It can be checked by a straightforward calculation, using that $-\frac{\Gamma(c)}{\Gamma(a)}y^{-c}\e^{y}$ is the Wronskian of $M$ and $U$, that
			\begin{equation}
	R(a;b;x,y) = -\frac{\Gamma(a)}{\Gamma(c)}y^{c-1}\e^{-y}\begin{cases}M(a,c;x)U(a,c;y) & \mbox{ for }x<y, \\
	M(a,c;y)U(a,c;x) & \mbox{ for }y<x.
	\end{cases}
	\end{equation}
        
	                We can transform the confluent operator as follows:
                        \begin{align}
                          &                    \frac{4}{x}{\e^{-\frac{x}{2}}}\big(x\pder_x^2+(c-x)\pder_x-a\big)\e^{\frac{x}{2}}\\ ={}&4\Big(\pder_x^2+\frac{c}{x}\pder_x+\big(\frac{c}{2}-a\big)\frac{1}{x}-      {\frac{1}{4}}\Big)\\={}&\partial_\rho^2+\frac{c}{\rho}\pder_\rho+\frac{c-2a}{\rho}-1,
                          \label{parto}\end{align}
                        where we substituted $x=2\rho$.
                        Let $\tilde R(a;c,\rho,\eta)$ denote the resolvent of (\ref{parto}). Then, using the fact that the kernel of an operator is a halfdensity in both variables, we obtain
                        \begin{align}&
                          \tilde R(a;c;\rho,\eta)\\
	                 ={} &\e^{-\rho}R(a;c;2\rho,2\eta)\e^\eta\eta\\
                  ={} &-\frac{\Gamma(a)}{\Gamma(c)}2^{c-1}{\eta^{c}}\e^{-\rho-\eta}
                            \begin{cases}M(a,c;2\rho)U(a,c;2\eta) & \mbox{ for }\rho<\eta, \\
	M(a,c;2\eta)U(a,c;2\rho) & \mbox{ for }\eta<\rho.
	                \end{cases}\end{align}

As we will see in the theorem below, the symbol of the resolvent of harmonic oscillator can be expressed by the Green's function of the confluent operator: 
	
		\begin{theorem}
			The symbol $F_{d,z}(\rho)$ has the representation
			\begin{equation}\label{whittaker-rep}
			\begin{split}
			F_{d,z}(\rho) =& 2^{d-1}\frac{\Gamma(\frac{d-z}{2})}{(d-1)!}\e^{-\rho}\Bigg( M\Big(\frac{d-z}{2},d;2\rho\Big)\int\limits_{\rho}^\infty\eta^{d-1}\e^{-\eta}U\Big(\frac{d-z}{2},d;2\eta\Big)\,\mathrm{d}\eta\\
			&+ U\Big(\frac{d-z}{2},d;2\rho\Big)\int\limits_{0}^{\rho}\eta^{d-1}\e^{-\eta}M\Big(\frac{d-z}{2},d;2\eta\Big)\,\mathrm{d}\eta \Bigg).
			\end{split}
			\end{equation}
		\end{theorem}
		
	\begin{proof}
		We would like to find the solution of $\one=(H-z)\Op(a_{d,z})$. Using the formula (\ref{moyal}) for  the Moyal star, we obtain
			\begin{equation}
			\begin{split}
			1 &= (x^2+p^2-z)*F_{d,z}(x^2+p^2) \\
			&= \big(x^2 + p^2 -z - \frac{1}{4}(\Delta_x^2 + \Delta_p^2)\big)F_{d,z}(x^2+p^2).
			\end{split}
			\end{equation}
			Substitution of $\rho=x^2+p^2$ leads us to the equation
			\begin{equation}
				\Big(-\pder_\rho^2 -\frac{d}{\rho}\pder_\rho -\frac{z}{\rho}+1\Big)F_{d,z}(\rho)=\frac{1}{\rho}.
			\end{equation}
	
	This leads us to 
	\begin{alignat}{4}
	F_{d,z}(\rho) &=&&-\int_0^\infty\tilde R\Big(\frac{d-z}{2};d;\rho,\eta\Big)\frac{1}{\eta}\mathrm{d}\eta\\ &=&&{}2^{d-1}\frac{\Gamma(\frac{d-z}{2})}{(d-1)!}\e^{-\rho}\Bigg( M\Big(\frac{d-z}{2},d;2\rho\Big)\int\limits_{\rho}^\infty\eta^{d-1}\e^{-\eta}U\Big(\frac{d-z}{2},d;2\eta\Big)\,\mathrm{d}\eta\\
	&&&+ U\Big(\frac{d-z}{2},d;2\rho\Big)\int\limits_{0}^{\rho}\eta^{d-1}\e^{-\eta}M\Big(\frac{d-z}{2},d;2\eta\Big)\,\mathrm{d}\eta \Bigg).
	\end{alignat}
        \end{proof}

        Recall that $M$ is analytic and its  asymptotic behaviour of $M$ around zero is $M(a,c;z)\sim 1$. However,
        for integer $c$ the function $U(a;c;z)$ is not analytic at $z=0$.
It can be written as
		\begin{equation}\label{K-at-integral}
		U(a,c;z) = \frac{(-1)^c}{\Gamma(a-c+1)\Gamma(c)}\Big(\log(z)M(a,c;z) +D(a,\,c;\,z)\Big),
		\end{equation}
		where $D(a,c; \,\cdot)$ is a meromorphic function of $z$ with a pole of order $c-1$.
                Thus, around $0$ the function $U(\frac{d-z}{2},d;x)$ diverges polynomially \emph{and} logarithmically, see \eqref{K-at-integral}. Hence,
\eqref{whittaker-rep} suggests that $F_{d,z}$ may have a logarithmic singularity at $0$.
However, we already know that  $F_{d,z}$ is analytic.

Let us check that      \eqref{whittaker-rep}  implies           the analyticy of $F_{d,z}$. Let us insert the representation \eqref{K-at-integral} of $U$ into \eqref{whittaker-rep} and split the integration range:
		\begin{alignat}{4}
		&&&(-1)^d 2^{-d+1}\cdot\frac{\Gamma(d)^2\Gamma(1 - \frac{d+z}{2})}{\Gamma(\frac{d-z}{2})} F_{d,z}(\rho)\nonumber\\
								   &=&& \e^{-\rho}M\big(\frac{d-z}{2},d;2\rho\big)\int\limits_{\rho}^1 \e^{-\eta}\eta^{d-1}\log\eta M\big(\frac{d-z}{2},d;2\eta\big)\,\mathrm{d}\eta\nonumber\\
								   &&&+ \e^{-\rho}M\big(\frac{d-z}{2},d;2\rho\big)\int\limits_{\rho}^{1}\e^{-\eta}\eta^{d-1}D\big(\frac{d-z}{2},\,d;\,2\eta\big)\,\mathrm{d}\eta \nonumber\\
			&&&+(-1)^d\Gamma(d)\Gamma(1-\frac{d+z}{2}) \e^{-\rho}M\big(\frac{d-z}{2},d;2\rho\big)\int\limits_{1}^\infty\eta^{d-1}\e^{-\eta}U\big(\frac{d-z}{2},d;2\eta\big)\,\mathrm{d}\eta \nonumber \\
								   &&&+ \e^{-\rho}\log\rho M\big(\frac{d-z}{2},d;2\rho\big)\int\limits_0^{\rho} \eta^{d-1}\e^{-\eta}M\big(\frac{d-z}{2},d;2\eta\big)\,\mathrm{d}\eta \nonumber\\
								   &&&+ \e^{-\rho}D\big(\frac{d-z}{2},\,d;\,2\rho\big)\int\limits_0^{\rho}\eta^{d-1}\e^{-\eta}M\big(\frac{d-z}{2},d;2\eta\big)\,\mathrm{d}\eta.
		\end{alignat}
	
After integrating the first integral by parts, the logarithms cancel out and we get the following result:
		\begin{alignat}{4}
		&&&(-1)^d 2^{-d+1}\cdot\frac{\Gamma(d)^2\Gamma(1 - \frac{d+z}{2})}{\Gamma(\frac{d-z}{2})} F_{d,z}(\rho)\nonumber\\
		&=&& -\e^{-\rho}M\big(\frac{d-z}{2},d;2\rho\big)\int\limits_\rho^1\frac{1}{\eta}\Big(\int\limits_0^\eta\e^{-r}r^{d-1}M\big(\frac{d-z}{2},d;2r\big)\,\mathrm{d}r\Big)\,\mathrm{d}\eta\nonumber\\
		&& &+ \e^{-\rho}M\big(\frac{d-z}{2},d;2\rho\big)\int\limits_{\rho}^{1}\e^{-\eta}\eta^{d-1}D\big(\frac{d-z}{2},d;2\eta\big)\,\mathrm{d}\eta \nonumber\\
		&&&+ (-1)^d\Gamma(d)\Gamma(1-\frac{d+z}{2})\e^{-\rho}M\big(\frac{d-z}{2},d;2\rho\big)\int\limits_{1}^\infty\eta^{d-1}\e^{-\eta}U\big(\frac{d-z}{2},d;2\eta\big)\,\mathrm{d}\eta \nonumber \\
		&&&+ \e^{-\rho}D\big(\frac{d-z}{2},d;2\rho\big)\int\limits_0^{\rho}\eta^{d-1}\e^{-\eta}M\big(\frac{d-z}{2},d;2\eta\big)\,\mathrm{d}\eta.
		\end{alignat}

		 Now it is easy to convince ourselves that $F_{d,z}$ is analytic in the region $\rho\sim0$. Let us start with analysing the first summand. The inner integral behaves at least like $\eta$ and so the outer integral does not give elements proportional to logarithm (or worse). The second summand consists of analytic functions times an integral behaving at zero at least like a constant. The third one is an analytic function times a finite number (the integral does not depend on $\rho$). The fourth summand is an integral behaving around zero like $\rho^{d}$, multiplied by analytic functions and $\rho^{-d+1}$, hence it is analytic as well.

                 \section{The Weyl symbol of spectral projections}

                 Recall that the eigenvalues of $H$ are $E_0:=d$, $E_1:=d+2$, $E_2:=d+4$,...           Let $P_n$ denote the spectral projection of $H$ onto $E_n$. In this section we will compute the Weyl symbol of $P_n$. We will see that it has a simple expression in terms of the Laguerre polynomial $L_n^{d-1}$. Thus we will reproduce a  recent result of Unterberger \cite{unterb2}, using the formula (\ref{whittaker-rep}) as the main tool.
                 \begin{theorem}
                   The Weyl symbol of $P_n$ is $p_n(x^2+p^2)$, where
                   \begin{equation}
                     p_n(\rho)=2^d(-1)^n
                     \e^{-\rho}L_n^{d-1}(2\rho).
                 \end{equation}\end{theorem}

                 \proof
                 The resolvent $(z-H)^{-1}$ has a pole at $E_n$ and the corresponding residue is $P_n$, that is
                 \begin{equation}
                   P_n=\Res(z-H)^{-1}\Big|_{z=E_n}.\end{equation}
                 Therefore,
                                  \begin{equation}
                   p_n(\rho)=-\Res F_{z,d}(\rho)\Big|_{z=E_n}.\end{equation}
                 All terms in  (\ref{whittaker-rep}) are analytic in $z$ around $E_n$ except for $\Gamma\big(\frac{d-z}{2}\big)$, which has a 1st order pole. By the well-known properties of the Gamma function
                 \begin{equation}
        -\Res\Gamma\big(\frac{d-z}{2}\big)\Big|_{z=E_n}=\frac{(-1)^n2}{n!}.
                 \end{equation}
                 Therefore,         		
			\begin{align*}\label{whittaker-rep1}
                          -\Res F_{d,z}(\rho)\Big|_{z=E_n} =& 
                          \frac{2^d(-1)^n}{n!(d-1)!}\e^{-\rho}\Bigg( M(-n,d;2\rho)\int\limits_{\rho}^\infty\eta^{d-1}\e^{-\eta}U(-n,d;2\eta)\,\mathrm{d}\eta\\			&+ U(-n,d;2\rho)\int\limits_{0}^{\rho}\eta^{d-1}\e^{-\eta}M(-n,d;2\eta)\,
                          \mathrm{d}\eta \Bigg).
			\end{align*}
                        Then we use
                        (\ref{laguerre1}) and (\ref{laguerre2}) which lead us to
                        $$p_n(\rho)=\frac{2^dn!}{(d+n-1)!}
                        \e^{-\rho}L_n^{d-1}(2\rho)\int_0^\infty\nu^{d-1}\e^{-\nu}
                        L_n^{d-1}(2\nu)\mathrm{d} \nu.$$
               The calculation that
               \begin{equation}
	               \int_0^\infty \nu^{d-1}\e^{-\nu}L_n^{d-1}(2\nu)\mathrm{d}\nu = (-1)^n\frac{(d+n-1)!}{n!}
               \end{equation}         
               is elementary after recalling that
               \begin{equation}
	               L^\alpha_n(x) = \frac{1}{n!}\e^x x^{-\alpha}\pder_x^n\e^{-x}x^{n+\alpha}.
               \end{equation}
                        
                         \qed

\section{The Weyl symbol of the inverse}
The inverse to harmonic oscillator is just the resolvent computed at the point $z=0$. Nevertheless, various properties of its symbol are easier than in the general case, which is the reason why we devote a separate section to the inverse.


	\begin{theorem}
		Putting $z=0$ in \eqref{integral-rep-of-res} we see that the function $F_{d,0}$ can be represented by
		\begin{equation}\label{integral-rep-of-inv}
			F_{d,0}(\rho) = \int\limits_{0}^{1}(1-t^2)^{\frac{d}{2}-1}\e^{-t\rho}\,\mathrm{d} t.
		\end{equation}
	\end{theorem}

	The integral representation \eqref{integral-rep-of-inv} easily implies the series representation of the symbol $F_{d,0}$:
	\begin{theorem}
		 $F_{d,0}$  can be represented by the following series:
		\begin{equation}\label{integral-rep-of-inv1}
			F_{d,0}(\rho ) = \frac{d!!}{d}\Big[ \alpha \sum\limits_{k=0}^\infty \frac{(2k-1)!!}{(2k+d-1)!!(2k)!}\rho^{2k} - \sum\limits_{k=1}^{\infty}\frac{(2k)!!}{(2k+d)!!(2k+1)!}\rho^{2k+1}\Big],
		\end{equation}
		where $\alpha=\frac{\pi}{2}$ for odd $d$ and $\alpha=1$ for even $d$.
	\end{theorem}
	\begin{proof}
		Write $F_{d,0}(\rho) = \sum\limits_{n=0}^{\infty} c_n \frac{\rho^n}{n!}$.	By 	\eqref{integral-rep-of-inv} we see that
		\begin{align}
			c_n &= \pder_\rho^n F_{d,0}(\rho)\Big|_{\rho=0} = (-1)^n\int\limits_0^{1}(1-t^2)^{\frac{d}{2}-1}t^n\,\mathrm{d}t\\
&=\frac{(-1)^n}{2}\int\limits_0^{1}(1-s)^{\frac{d}{2}-1}s^{\frac{n+1}{2}-1}\,\mathrm{d}s = \frac{(-1)^n}{2}\frac{\Gamma(\frac{d}{2})\Gamma(\frac{n+1}{2})}{\Gamma(\frac{d+n+1}{2})},
		\end{align}
                	where we set $t^2=s$ 
and applied the integral formula for Euler's Beta function. The relation $\frac{\Gamma(x+1)\Gamma(y)}{\Gamma(x+y+1)} = \frac{x}{x+y}\frac{\Gamma(x)\Gamma(y)}{\Gamma(x+y)}$ leads to an easy recurrence, which gives for $k\in\bbN$
			\begin{eqnarray}
				c_{2k+1} &=& -\frac{1}{2}\cdot\frac{(2k)!! d!!}{(2k+d)!!}\frac{\Gamma(\frac{d}{2})}{\Gamma(\frac{d}{2}+1)} = -\frac{d!!}{d}\cdot\frac{(2k)!!}{(2k+d)!!},\\
				c_{2k} &=& \frac{1}{2}\cdot\frac{(2k-1)!!(d-1)!!}{(2k+d-1)!!}\frac{\Gamma(\frac{1}{2})\Gamma(\frac{d}{2})}{\Gamma(\frac{1+d}{2})} = \begin{cases} \frac{\pi}{2}  \frac{(2k-1)!!}{(2k+d-1)!!} &\mbox{ for }d\mbox{ odd,} \\  \ \ \frac{(2k-1)!!}{(2k+d-1)!!} &\mbox{ for }d\mbox{ even.}
				\end{cases}
			\end{eqnarray}
		\end{proof}
Note that the representation (\ref{integral-rep-of-inv1}) can be found in \cite{toft}, proven by a different method.

        \section{Bessel-type functions  and the harmonic oscillator}
	The \emph{modified Bessel function} $I_m$ and the \emph{MacDonald function} $K_m$, both solutions of the modified Bessel equation,
        can be easily represented by confluent functions $M$ and $U$:
		\begin{equation}
		\begin{split}
		I_m(z) &= \frac{(\frac{1}{2}z)^m\e^{-z}}{\Gamma(1+m)}M(m+\frac{1}{2},2m+1;2z),\\
		K_m(z) &= \pi^{\frac{1}{2}}(2z)^m\e^{-z}U(m+\frac{1}{2},2m+1;2z).
		\end{split}
		\end{equation}
		Setting $z=0$ in \eqref{whittaker-rep} and
                using
                $\frac{\Gamma(\frac{d}{2})}{\Gamma(d)}= \frac{\sqrt{\pi}2^{1-d}}{\Gamma(\frac{d+1}{2})}$,
we can express the inverse of the harmonic oscillator in terms of Bessel-type functions.
		\begin{theorem}
Set $ m =\frac{d-1}{2}$.		The function $F_{d,0}$ can be represented as
			\begin{equation}
			F_{d,0}(\rho) = \rho^{-m}\Bigg(K_m(\rho)\int\limits_{0}^{\rho}\eta^mI_m(\eta)\,\mathrm{d} \eta  + I_m(\rho)\int\limits_{\rho}^{\infty}\eta^mK_m(\eta)\,\mathrm{d}\eta\Bigg).
		\label{bessel}	\end{equation}
		\end{theorem}

	        Note that the authors of \cite{toft} ask the question whether it is possible to express the Weyl symbol of the inverse of the harmonic oscillator in terms of known special functions. (\ref{bessel}) is our answer to this question.

                \section{Representation by elementary functions\\ for even dimensions}
                                It is well-known that Bessel functions of odd parameters can be expressed in terms of elementary functions. For instance,
                                $I_\frac{1}{2}(x) = \sqrt{\frac{2}{\pi x}}\sinh x$ and $K_\frac{1}{2}(x) = \sqrt{\frac{\pi}{2x}}\e^{-x}$. Therefore,
                                			\begin{equation}
			F_{2,0}(\rho) = \frac{1 - \e^{-\rho}}{\rho},
			\end{equation}		                                which has also been calculated in \cite{toft} (see equation (0.3) in there).  An even easier derivation of this fact is to put $d=2$ in \eqref{integral-rep-of-inv} and integrate.

More generally,
for even $d$ the symbol of the inverse, $F_{d,0}$, can  be represented by elementary functions.

\begin{theorem}
For $d=2p$, $p=1,2,\dots$, we get
		\begin{equation}
		\begin{split}
			F_{d,0}(\rho) 		  &= \mathop{\sum}\limits_{k=0}^{p-1} {{p-1}\choose{k}}(-1)^k(2k)!\rho^{-2k-1} \Big( 1 - p_{2k}(\rho)\e^{-\rho}\Big),\label{explicit}
		\end{split}
		\end{equation}
		where $p_j(t) = \mathop{\sum}\limits_{l=0}^{j}\frac{t^l}{l!}$ is the partial Taylor expansion of $\e^t$.\end{theorem}

\begin{proof} We compute:
		\begin{equation}
		\begin{split}
			F_{d,0}(\rho) &= \int\limits_{0}^{1}(1-t^2)^{p-1}\e^{-t\rho}\,\mathrm{d} t = \mathop{\sum}\limits_{k=0}^{p-1} {{p-1}\choose{k}}(-1)^k\rho^{-2k-1}\int\limits_0^\rho v^{2k} \e^{-v}\,\mathrm{d}v\\
							\end{split}
		\end{equation}
and then we use Lemma \ref{below}. 
\end{proof}

\begin{lemma}\label{below}
  For $n=0,1,\dots,$ we have
  \begin{equation}
    \int_0^\rho v^n\e^{-v}\mathrm{d} v=(-1)^nn!\big(p_n(\rho)\e^{-\rho}-1\big).\label{below1}\end{equation}\end{lemma}

\begin{proof}
For $n=0$, (\ref{below1}) is immediate. Then we use induction wrt $n$:  
		\begin{equation}
			\int\limits_0^\rho v^{n} \e^{-v}\,\mathrm{d}v = -\rho^{n}\e^{-\rho} + n\int\limits_0^\rho v^{n-1} \e^{-v}\,\mathrm{d}v.
		\end{equation}
\end{proof}
(\ref{explicit}) was proven by a different method  in \cite{toft}.

		\section{Bounds on derivatives of the symbol}

                The  symbol of the resolvent is a very well behaved smooth function. Its derivatives satisfy  bounds, which we describe in the theorem below.           They follow easily from the integral representation \eqref{integral-rep-of-res}.
		\begin{theorem}
			Let $\Re(z)\leqslant 0$, $n=0,1,2,\dots$ and let $\alpha$ be a multi-index. Let $0\leq s\leq1$. Then the following bounds are true:
			\begin{eqnarray}
				|\pder_\rho^n F_{d,z}(\rho)| &\leqslant&   (n!)^s(n+1)^{s-1} \rho^{-s(n+1)},\quad d\ge2; \label{bound-of-F}\\
|\pder_\rho^n F_{1,z}(\rho)| &\leqslant&  C (n!)^s(n+1)^{\frac{s}{2}-\frac12} \rho^{-s(n+1)}.\label{bound-of-F-1}
                        \end{eqnarray}
			where $C$ is a constant independent of  $n$.
\label{estimate}\end{theorem}
		\begin{proof}
                 The case $d\geq2$ is very
                  easy. From \eqref{integral-rep-of-res} we get
			\begin{equation}\label{bound-rho}
				|\pder_\rho^n F_{d,z}(\rho)| \leqslant \int\limits_0^1 (1-t^2)^{\frac{d}{2}-1}t^n\e^{-t\rho}\,\mathrm{d}t \leqslant \int\limits_0^\infty t^n\e^{-t\rho}\,\mathrm{d}t = n! \rho^{-n-1}.
			\end{equation}
Next we estimate
	\begin{equation}\label{bound-rho.1}
				|\pder_\rho^n F_{d,z}(\rho)| \leqslant \int\limits_0^1 (1-t^2)^{\frac{d}{2}-1}t^n\e^{-t\rho}\,\mathrm{d}t \leqslant \int\limits_0^{1} t^n\,\mathrm{d}t = \frac{1}{n+1}.
			\end{equation}
        Then we interpolate between (\ref{bound-rho}) and (\ref{bound-rho.1})
        to obtain (\ref{bound-of-F}).

        The case $d=1$ is slightly more complicated.                        
			\begin{align}\label{bound-rho1}
			  |\pder_\rho^n F_{1,z}(\rho)|& \leqslant \int\limits_0^1 (1-t^2)^{-\frac{1}{2}}t^n\e^{-t\rho}\,\mathrm{d}t\\
                          &
                          \leqslant \int\limits_0^{\frac12} (1-t^2)^{-\frac{1}{2}}t^n\e^{-t\rho}\,\mathrm{d}t
                          +\int\limits_{\frac12}^1 (1-t^2)^{-\frac{1}{2}}t^n\e^{-t\rho}\,\mathrm{d}t\\
&
                          \leqslant \Big(1-\frac{1}{2^2}\Big)^{-\frac12}\int\limits_0^\infty t^n\e^{-t\rho}\,\mathrm{d}t
                          +\int\limits_{\frac12}^1 (1-t^2)^{-\frac{1}{2}}\e^{-\frac{\rho}{2}}\,\mathrm{d}t                       \\
                          &=C_1n!\rho^{-n-1}+C_2\e^{-\frac{\rho}{2}}\leq C n!\rho^{-n-1}.\label{align}
                        \end{align}

        Next,  consider $0<\epsilon<1$:
        		\begin{align}\label{bound-rho2}
			  |\pder_\rho^n F_{1,z}(\rho)|& \leqslant \int\limits_0^1 (1-t^2)^{-\frac{1}{2}}t^n\,\mathrm{d}t\\
         &\leq\int_0^{1-\epsilon}(2\epsilon-\epsilon^2)^{-\frac12}t^n\mathrm{d} t
                          +
                          \int_{1-\epsilon}^1(1-t^2)^{-\frac12}\mathrm{d} t     \\
                          &\leq
                          \frac{C}{2\sqrt{\epsilon}(n+1)}+\frac{C\sqrt{\epsilon}}{2}
                          .\label{bound+}\end{align}
                        Setting $\epsilon:=\frac{1}{\sqrt{n+1}}$ we obtain
                        \begin{equation}
                        |\pder_\rho^n F_{1,z}(\rho)| \leqslant\frac{C}{\sqrt{n+1}}.\label{bound-}\end{equation}
Interpolating between (\ref{align}) and (\ref{bound-}) we obtain
(\ref{bound-of-F-1}).         					\end{proof}

                Note that the main result of  \cite{toft} are
the estimates                (0.6), (0.7)  and (0.8) 
for the symbol of the inverse of the harmonic oscillator  $a_{d,0}(x,p)$.
             The authors call them  
   \emph{bounds of Gelfand--Shilov type}.
       Our      Thm  \ref{estimate} easily implies  the bounds 
  (0.6), (0.7)  and (0.8) of \cite{toft}.

\end{document}